%% file: Sparse_Actuator_Schedule_of_LDS.tex
\documentclass[letterpaper, 10 pt, conference]{ieeeconf}
\IEEEoverridecommandlockouts                             
\usepackage{amsmath,mathtools,amssymb}
\let\labelindent\relax

\usepackage{algorithm,algpseudocode,algorithmicx}
\usepackage{cite,cleveref,enumitem}
\usepackage{caption,subcaption}
\usepackage{flushend}
\usepackage{xcolor}
\usepackage{url}

\usepackage{graphicx}					
\graphicspath{{figures/}}

\input{My_notations}

\newcommand*{\qedb}{\null\nobreak\hfill\ensuremath{\blacksquare}}

\algdef{SE}[DOWHILE]{Do}{doWhile}{\algorithmicdo}[1]{\algorithmicwhile\ #1}%

\Crefname{figure}{Fig.}{Figs.}
\Crefname{table}{Tab.}{Tabs.}
\Crefname{section}{Sec.}{Secs.}

\title{Sparse Actuator Scheduling for Discrete-Time Linear Dynamical Systems}

\author{Krishna Praveen V. S. Kondapi$^{1}$, Chandrasekhar Sriram$^{2}$, Geethu Joseph$^{3}$ and Chandra R. Murthy$^{1}$ \\ {\it Fellow, IEEE} 
\thanks{$^{1}$K. Kondapi and C. R. Murthy are with the Dept. of ECE at the Indian Institute of Science (IISc), Bangalore 560012, India. Emails:
        {\tt\small \{praveenkvsk, cmurthy\}@iisc.ac.in}.}%
        \thanks{$^{2}$C. Sriram is with Texas Instruments India Pvt. Ltd. He was at the Dept. of ECE, IISc, during the course of this work. Email: {\tt\small chandrasekhars@alum.iisc.ac.in}}
\thanks{$^{3}$G. Joseph is with the Faculty of Electrical Engineering, Mathematics, and Computer Science at the Delft University of Technology, Delft 2628 CD,  Netherlands. Email:
        {\tt\small G.Joseph@tudelft.nl}.}
}

\begin{document}

\maketitle
\thispagestyle{empty}
\pagestyle{empty}

%%%%%%%%%%%%%%%%%%%%%%%%%%%%%%%%%%%%%%%%%%%%%%%%%%%%%%%%%%%%%%%%%%%%%%%%%%%%%%%%
\begin{abstract}
We consider the control of discrete-time linear dynamical systems using sparse inputs where we limit the number of active actuators at every time step. We develop an algorithm for determining a sparse actuator schedule that ensures the existence of a sparse control input sequence, following the schedule, that takes the system from any given initial state to any desired final state. Since such an actuator schedule is not unique, we look for a schedule that minimizes the energy of sparse inputs. For this, we optimize the trace of the inverse of the resulting controllability Gramian, which is an approximate measure of the average energy of the inputs. We present a greedy algorithm along with its theoretical guarantees. Finally, we empirically show that our greedy algorithm ensures the controllability of the linear system with a small number of active actuators per time step without a significant average energy expenditure compared to the fully actuated system.
%a better performance in the low sparsity compared to other schedulers at low sparsity and similar performance at high sparsity.

\end{abstract}

%%%%%%%%%%%%%%%%%%%%%%%%%%%%%%%%%%%%%%%%%%%%%%%%%%%%%%%%%%%%%%%%%%%%%%%%%%%%%%%%

\section{Introduction}
\label{sec:intro}
The field of networked control systems, in which the controllers, sensors, and actuators communicate over a band-limited network to achieve a given control objective, is a growing area of research \cite{control_comm_constraint, jadhbabaie2019, NAGAHARA201184}. Due to the bandwidth constraints, these systems demand communication-efficient control inputs. One way to reduce the communication cost is to use only a few available actuators at any given time instant, i.e.,  the number of nonzero entries of the control input is small compared to its length, leading to a \emph{sparse control input}~\cite{Li_2016_CSinSwitchedSys,NAGAHARA201184}. The sparse vectors admit compact representation in a suitable basis~\cite{foucart2013math}, and thus, save bandwidth. Sparse control also finds applications in resource-constrained systems such as environmental control systems and network opinion dynamics \cite{joseph2020controllability, WENDT201937}. Motivated by this, we consider the problem of designing an actuator schedule for sparse control inputs to drive the system to any given desired state.

The controllability of a discrete-time linear dynamical system (LDS) using sparse inputs, along with its necessary and sufficient conditions, was defined and developed in \cite{gjoseph2020controllability}.  Here, at each time step, the input can utilize at most $s$ actuators, but the selected subset of $s$ actuators can be time-varying. However, the work does not provide the actuator schedule or design the sparse inputs. Some researchers have considered the actuator scheduling for continuous-time systems by maximizing the trace of the controllability Gramian with a limit on the number of active actuators per time step~\cite{Olshvesky_2020_ActRlx,Ikeda_2018_TVNodeSlct,Fotiadis_2021_ActPlctLearn,Fotiadis_2022_DataBasedActSlct}. This problem has been extended to cases where system dynamics are unknown, resulting in the development of an online, data-driven approach for learning the optimal actuator placement~\cite{Fotiadis_2021_ActPlctLearn,Fotiadis_2022_DataBasedActSlct}. Some works have considered a restrictive variant where the subset of chosen actuators remains fixed across time~\cite{Olshevesky_2014_MinControl,Tzoumas_2016_CtrlEffort,Dilip_2019_CtrlGramOptSlct}. However, time-varying actuator scheduling can potentially ensure the controllability of systems that are uncontrollable systems when using inputs with fixed support, and could guide the LDS to the desired state faster than time-invariant actuator scheduling. A variant of time-varying sparse actuator scheduling problem considered in the literature only limits the \emph{average} number of active actuators across all the time steps~\cite{jadbabaie2018deterministic,SIAMI2020109054}. They developed schedulers that yield a Gramian matrix that approximates the Gramian matrix of the fully actuated system. An alternative approach for the time-varying actuator scheduling problem is using a linear-quadratic regulator or linear-quadratic-Gaussian control~\cite{Chamon_2019_MatroidOpt,Jiao_2023_ActAchCvxRlx}. The authors in~\cite{Zhao_2016_NodeSchControl} also consider a time-varying sparse schedule with the goal of obtaining controllability with a symmetric transfer matrix $\matA$. 
%These approaches does not guarantee the controllability of the system.} \textcolor{red}{These two sentences seem to contradict each other. If \cite{7798535} aims to achieve controllability, why don't they meet that aim?} 
To the best of our knowledge, our problem of identifying a time-varying sparse actuator schedule that ensures controllability with minimal control energy has not been explored in the literature.
%We use some of the ideas in this work to find an actuator schedule such that the system is $s$-sparse controllable. To summarize, sparse control with inputs having dynamic support has not been well studied in the literature, to the best of our knowledge. 

We formulate our actuator scheduling problem as minimizing the trace of the inverse of the controllability Gramian, a widely used metric for quantifying average control enegry~\cite{jadbabaie2018deterministic}. This optimization is subject to a constraint limiting the active actuators to at most $s$ per time step. The specific contributions of this paper are as follows:
\begin{itemize}
	\item We show that our objective function exhibits approximate supermodularity, and the sparsity constraint is a matroid. This result leads to a greedy algorithm for finding the sparse actuator schedule with provable guarantees.
    \item Using our algorithm, we empirically study the tradeoff between the sparsity and average energy for time-invariant and time-varying schedules. We observe that, for large random systems with sparse and time-varying schedules, the increase in the average control energy compared to the average energy required by a fully actuated system is inversely proportional to the fraction of active actuators. %This observation is supported by an intuitive explanation grounded in matrix theory.
\end{itemize}
Overall, we present a systematic study of sparse time-varying actuator scheduling, develop a novel scheduling algorithm with theoretical guarantees, and empirically elucidate its energy requirements.

\section{Actuator Scheduling Problem}
\label{sec:sys_model}
We consider a discrete-time LDS whose dynamics are governed by the following equation:
\begin{equation}
\vecx (k+1) = \matA \vecx (k) + \matB \vecu(k), \label{eq:state}
\end{equation}
where $k = 0,1,\ldots$ is the integer time index. Here, $\vecx (k) \in \mathbb{R}^n $ is the system state and $\vecu(k) \in \mathbb{R}^m$ is the input to the system. %, and $\vecy(k) \in \mathbb{R}^p$ is the output of the LDS. 
The matrices $\matA \in \mathbb{R}^{n \times n}$ and $\matB \in \mathbb{R}^{n \times m}$
%, and  $\matC \in \mathbb{R}^{p \times n}$ 
are the system transfer matrix and the input matrix, 
%and the input matrix 
respectively. We aim to design $s$-sparse inputs, i.e., $\Vert\vecu(k)\Vert_0\leq s$, for all values of $k$ to drive the system to a desired state. 
Finding sparse inputs is hard because of the nonconvex sparsity constraint. However, if the support $\calS_k$ of $\vecu_k$ is known, we can easily find the sparse vectors using techniques like the least squares method. Therefore, we aim to design a sparse actuator schedule (support of the sparse inputs), independent of the initial and final state of the LDS,  that can drive the LDS to any given desired state in $K$ time steps. %\textcolor{black}{In the sequel, we term this as a \emph{state-independent} actuator schedule since the schedule is independent of the initial and final states.} 

We denote the sparse actuator schedule by $(\calS_0,\calS_1,\ldots,\calS_{K-1})$, where, $\calS_k\subseteq \{1,2,\ldots,m\}$ and $ \vert \calS_k \vert \leq s$ for all values of $k$ to satisfy the sparsity constraint. \textcolor{black}{ This sparsity structure is referred to as piecewise sparsity.}
%This results in $s$-sparse inputs i.e., $\Vert\vecu_k\Vert_0\leq s$, for all values of $k$. The conditions and algorithm to find s-sparse control inputs to hold the system in desired are discussed in ~\cite{icassp_paper}.
It is known that for any  $s$-sparse controllable LDS,
%for any initial and final states, there exists a sparse input sequence $\vecu(k)$ for  $k = 0, 1, \ldots, K-1<\infty$ with $\Vert\vecu(k)\Vert_0\leq s$, which steers the system from the state $\vecx(0) = \vecx_0$ to $\vecx(K)=\vecx_f$. 
% From \eqref{eq:state}, the condition $\vecx(K)=\vecx_f$ leads to
% \begin{equation} \label{eq:control_eq}
% 	\vecx(K) - \matA^K \vecx_0 = \matR_{(K)} \tilde{\vecu}_{(K)},
% \end{equation}
% where we define the controllability matrix and joint input vector as
% \begin{align}
% \matR_{(K)} &\triangleq \begin{bmatrix}
%  \matA^{K-1} \matB & \matA^{K-2} \matB & \ldots & \matB 
% \end{bmatrix} \in \mathbb{R}^{n \times mK}\\
%  \tilde{\vecu}(K) &\triangleq 
%  \begin{bmatrix}
%    \vecu (0)^{\T} & \vecu (1)^{\T} \ldots \vecu (K-1)^{\T} 
% \end{bmatrix}
% ^{\T} \in \mathbb{R}^{mK}.\label{eq:defn_tildeu}
% \end{align}
we can find a finite $K$ and sparse actuator schedule $(\calS_0,\calS_1,\ldots,\calS_{K-1})$ with $\vert \calS_k \vert \leq s$ such that
%$\rank{ \matR_{(K)}^{\calS}} = n$, where 
% \begin{equation}
% \matR_{(K)} \triangleq \begin{bmatrix}
%  \matA^{K-1} \matB & \matA^{K-2} \matB & \ldots & \matB 
% \end{bmatrix} \in \mathbb{R}^{n \times mK}.
% %  \tilde{\vecu}(K) &\triangleq 
% %  \begin{bmatrix}
% %    \vecu (0)^{\T} & \vecu (1)^{\T} \ldots \vecu (K-1)^{\T} 
% % \end{bmatrix}
% \end{equation}
\begin{equation} \label{eq:control_con}
\rank{\begin{bmatrix} 
\matA^{K-1} \matB_{\calS_0} & \matA^{K-2} \matB_{\calS_1} & \ldots & \matB_{\calS_{K-1}} 
\end{bmatrix}}=n.
\end{equation}
Sparse controllability is equivalent to two conditions being satisfied: the system is controllable and $s\geq n-\rank{\matA}$~\cite{gjoseph2020controllability}.
So, in the sequel, to ensure that the problem is feasible, we assume that these two conditions hold. 

%We note that the actuator schedule $\calS$, i.e., the support sequence of the sparse inputs is independent of the initial and final state of the LDS. That is, there exists a \emph{single} sequence of sparse supports that can drive the system from \emph{any} initial state to \emph{any} final state in $K$ steps. Hence, we aim to find a sparse actuator schedule $\calS$ for a given sparse controllable LDS that achieves $\rank{ \matR_{(K)}^{\calS}} = n$. 

% \section{Driving to a Desired State} \label{sec:greedy_sch}`
% We discuss an approach to steer an $s$-sparse controllable LDS to a desired state using $s$-sparse inputs.
% We design sparse inputs such that the actuator schedule is fixed and does not change with the initial and final states.

% \subsection{State-independent Actuator Scheduling}
% \label{sec:state_independent}
% The goal of this section is to find an {unweighted} actuator schedule or, equivalently, to find the  set  $\calS = \{\calS_i\subseteq \{1,2,\ldots,m\}: \vert \calS_i \vert \leq s\}_{i=0}^{K-1}$  for which $\rank{ \matR_{(K)}^{\calS}} = n$. Using the actuator schedule $\calS$, the nonzero part of sparse inputs can be obtained as 
% \begin{equation}\label{eq:nonzero_input}
% \tilde{\vecu}_{(K),\calS}=\lb\matR_{(K)}^{\calS}\rb^{\T}\ls \matR_{(K)}^{\calS}\lb\matR_{(K)}^{\calS}\rb^{\T}\rs^{-1}(\vecx_f - \matA^K \vecx_0).
% \end{equation}

Further, for an $s$-sparse controllable LDS, the minimum number of time steps $K$ required to ensure controllability is bounded as~\cite{gjoseph2020controllability}  
\begin{equation} \label{eq:sparse_cntrl_bounds}
\frac{n}{s} \leq K \leq \min\lb q \left\lceil \frac{\rank{\matB}}{s} \right\rceil, n-s+1 \rb \leq n.
\end{equation}
where $q$ is the degree of the minimum polynomial of $\matA$. For simplicity, we let $K=n$.
% From \eqref{eq:sparse_cntrl_bounds}, we note that there exists an actuator schedule of length $K\leq n$ for which $\rank{ \matR_{(K)}^{\calS}} = n$. So for ease of exposition, we fix $K=n$. 
%However, our approach directly applies to any $K$ for which a schedule exists such that \eqref{eq:control_con} holds. 
Our goal is to find an actuator schedule $\calS = (\calS_0,\calS_1,\ldots,\calS_{n-1})\in\Phi$ such that \eqref{eq:control_con} holds, where the feasible set $\Phi$ for $\mathcal{S}$ is defined as
\begin{equation}\label{eq:feasible_set}
    \Phi \!=\! \lc\!(\calS_0,\calS_1,\ldots,\calS_{n-1})\!\!:\calS_k\!\subseteq \!\{1,2,\ldots,m\}, \vert \calS_k \vert \leq s,\forall k\rc.
\end{equation}

The rank of the controllability matrix can be analyzed using the Gramian matrix. For any actuator schedule $\calS = (\calS_0,\calS_1,\ldots,\calS_{n-1})\in\Phi$, the corresponding Gramian  is
\begin{equation}\label{eq:W_mat_Defn}
	\matW_{\calS} = \sum_{k=1}^{n} \matA^{k-1} \matB_{\calS_{n-k}} \matB_{\calS_{n-k}}^{\T} (\matA^{k-1})^{\T}.
\end{equation}
%\sout{We note that $\matW_{\calS}$ is a submatrix of the unconstrained Grammian,} \textcolor{red}{how is it a submatrix?}. 
Our goal is to select the actuator schedule $\calS$ so that $\rank{\matW_{\calS}}=n$, which in turn ensures that the system is controllable. However, the solution does not need to be unique. Therefore, we look at minimizing $\trace{\matW_{\calS}^{-1}}$ which is a popular measure of the ``difficulty of controllability''~\cite{Olshevsky_2018_NonSupermodular}. It represents the average energy to drive the LDS from $\vecx(0)=\zero$ to a uniformly random point on the unit sphere. % i.e., $\mathbb{E} \lc \vecx_f \vecx_f^{\T}\rc = \eye$. From \eqref{eq:nonzero_input}, we have
% \begin{align} 
% \mathbb{E}
% \lc \tilde{\vecu}_{(n)}^{\T} \tilde{\vecu}_{(n)}\rc
%   & = \trace{ \lb\matR_{(n)}^{\calS}\rb^{\T} \matW_{\calS}^{-1} \mathbb{E} \lc\vecx_f \vecx_f^{\T}\rc \matW_{\calS}^{-1} \matR_{(n)}^{\calS}}\\&=\trace{ \matW_{\calS}^{-1}}.
% \end{align}
% Also, $\trace{ \matW_{\calS}^{-1}}$ is finite only if $\matW_{\calS}$ is full rank. Therefore, our scheduling problem is formulated as
% \begin{equation}\label{eq:actuatorschedule}
%     \underset{\calS \in\Phi}{\arg\min} \trace{\matW_\calS^{-1}}.
% \end{equation}
% where $\Phi$ is defined in \eqref{eq:feasible_set} and $\matW_\calS$ is given by \eqref{eq:W_mat_Defn}. The above combinatorial optimization problem is difficult to solve exactly. So, 
We present a greedy algorithm to minimize~$\trace{ \matW_{\calS}^{-1}}$ next.

\section{Greedy Scheduling Algorithm} \label{sec:greedy_sch}
%One approach to solve \eqref{eq:opt_energy_epsilon} is to use a greedy approach. 
First, we note that $\trace{ \matW_{\calS}^{-1}}$ is not well-defined if $\matW_\calS$ is rank-deficient. So, we 
%we consider an approximation of the optimization problem in \eqref{eq:actuatorschedule}. We 
use the lower bound $\trace{(\matW_{\calS} + \epsilon \matI)^{-1}} (\leq \trace{(\matW_\calS)^{-1}})$ as an alternative cost function, referred to as the $\epsilon$-auxiliary energy, for some small $\epsilon>0$. The term $\epsilon \matI$ guarantees that the inverse exists and the cost function is well-defined for any schedule. Thus, our new optimization problem is
\begin{equation}
 \label{eq:opt_energy_epsilon}
    \underset{\calS \in\Phi}{\arg\min} \; \trace{ (\matW_\calS + \epsilon \matI)^{-1} }.
\end{equation}
where $\Phi$ is defined in \eqref{eq:feasible_set} and $\matW_\calS$ is given by \eqref{eq:W_mat_Defn}. As $\epsilon$ decreases, \eqref{eq:opt_energy_epsilon} becomes a better approximation of the original objective of attaining full rank matrix $\matW_\calS$. %More concretely, we have the following result from~\cite{7122316}. 
% \begin{lemma} \label{lma:controllabl_constr_sat}
% For the matrix $\matW_\calS$ defined in \eqref{eq:W_mat_Defn}, if the condition $\trace{(\matW_S + \epsilon \matI)^{-1}} \leq \frac{1}{\epsilon}$ holds, then $\rank{\matW_S} = n$.
% \end{lemma}

%We use the following approach to minimize $\trace{ \matW_{\calS}^{-1}}$. First, we initialize $\mathcal{S}_k=\emptyset$ for all values of $k$. Then, at each iteration, we include an element in one of the sets that minimizes the cost function. 
Next, we observe that any actuator schedule of length $n$ can be represented using a subset of the set 
\begin{equation}\label{eq:V_defn}
    \calV = \{ (k,j) \vert k=0,1,\ldots,n-1, \;j=1,2,\ldots,m \}.
\end{equation}
Here, $k$ represents the time index, and $j$ represents the actuator index. Let $2^\calV$ denote the power set of $\calV$. It is clear that there is a natural bijection between any $\calT\in 2^\calV$ and the corresponding actuator schedule:
\begin{equation}\label{eq:S2Tmapping}
    \calS(\calT) = (\calS_0,\calS_1,\ldots,\calS_{n-1})\;\;\text{with}\;\; \calS_k=\{j: (k,j)\in\calT\}.
\end{equation}
Therefore, the problem in \eqref{eq:opt_energy_epsilon} can be written using $\calV$ as 
\begin{equation}
 \label{eq:opt_energy_epsilon_mod}
\underset{\calT\in  2^\calV}\min \; \trace { (\matW_{\calS(\calT)} + \epsilon \matI)^{-1} }
\;\text{s.t.} \;\calS(\calT)\in\Phi.
\end{equation}

%To this end, let $2^\calV$ denote the power set of $\calV$ and define the function $E: 2^\calV\times\bb{R_{++}} \to \bb{R}$ as $E(\calT,\eps) = - \trace { (\matW_{\calS(\calT)} + \epsilon \matI)^{-1} }$. 
%with  the actuator schedule $\calS(\calT)$ satisfying the sparsity constraints. 
%For any $\calT\in 2^\calV$, the corresponding actuator schedule is
%\begin{equation}\label{eq:S2Tmapping}
%    \calS(\calT) = (\calS_0,\calS_1,\ldots,\calS_{n-1})\;\;\text{with}\;\; \calS_k=\{j: (k,j)\in\calT\}.
%\end{equation}
%Hence, \eqref{eq:opt_energy_epsilon} is equivalent to \textcolor{red}{why -Tr here?}
%\begin{equation}
% \label{eq:opt_energy_epsilon_mod}
%\underset{\calT\in  2^\calV}\max \; E(\calT,\epsilon) = - \trace { (\matW_{\calS(\calT)} + \epsilon \matI)^{-1} }
%\;\text{s.t.} \;\calS(\calT)\in\Phi.
%\end{equation}
%where $\Phi$, $\calV$, and $\calS(\calT)$ are defined in \eqref{eq:feasible_set}, \eqref{eq:V_defn}, and \eqref{eq:S2Tmapping}, respectively.

Using the new formulation, the greedy algorithm starts with $\mathcal{T}$ being the empty set and finds the element from $\calV$ which when added to $\calT$ minimizes the cost function. 
%adds an element from $\calV$ to the set $\calT$, leading to the minimum cost function in each iteration. 
Specifically, in the $r$th iteration of the algorithm, let $\calT^{(r)}$ set of indices collected up to the previous iteration. Then, we find
\begin{multline}\label{eq:greedy_update}
    (k^*,j^*) = \underset{(k,j)\in  \calV^{(r)}}{\arg\min} \trace { (\matW_{\calS(\calT^{(r)}\cup\{(k,j)\})} + {\epsilon} \matI)^{-1} },
\end{multline}
where $\calV^{(r)} \subseteq \calV\setminus\calT^{(r)}$ is obtained by removing all infeasible index pairs, i.e., 
\begin{equation} \label{eq:search_space}
    \calV^{(r)} = \left\{ (i,j) \in \calV \setminus \calT^{(r)}: \calS(\calT^{(r)} \cup (i,j)) \in \Phi \right\}.
\end{equation}

%\textcolor{red}{where $\calV^{(r)} \in 2^{\calV}$ obtained by removing $\calT^{(r)}$ and all the infeasible points in the subsequent iterations, i.e., 
%\begin{equation}
%   \lc (k,j)\in\calV: \calS(\calT^{(r-1)})=(\calS_1,\calS_2,\ldots,\calS_n),\;\;\lv\calS_k\rv=s\rc.
%\end{equation}}
 
Finally, to complete the algorithm's description, we must choose $\epsilon$. We start with some $\epsilon_0>0$ and repeat the greedy algorithm (outer loop) for decreasing values of $\epsilon$ until we obtain a full-rank Gramian matrix. The overall procedure is outlined in \Cref{alg:greedy_algo_epsilon}. %\textcolor{black}{Note that \Cref{alg:greedy_algo_epsilon} returns a sparse actuator schedule $\calS$ that is independent of the initial and final state}. 

%\textcolor{red}{It makes no sense to write the blue sentence at the end of the previous paragraph because it is completely unrelated to anything else that is described here. The rule of thumb is simple: Each paragraph should convey ONE idea. The previous paragraph's goal is to say that we can use the approach outlined above to develop an iterative algorithm that starts at some epsilon and keeps decreasing epsilon. Also, it points to Algorithm 1.}

%\textcolor{red}{To describe the ``state-independent'' part of the title of Algorithm 1, go back to where you are describing the problem statement, and find the right place to insert a comment like ``we wish to determine a sparse actuator schedule that can be determined independent of the initial and final states; this will be referred to as a \emph{state-independent} sparse activator schedule in the sequel.}

%\textcolor{red}{I have done this for you, please see the sentence in red in Sec. II. Then, you can comment out all these red comments, the blue sentence, and change the red bit in Sec. II to black (after you have read it carefully and corrected any errors you find in it.)} % \textcolor{black}{State-independent}

\begin{algorithm}
	\caption{Greedy actuator scheduling} 
	\begin{algorithmic} [1]
		\Require System matrices: $\matA,\matB$; sparsity $s$
		\Statex \hspace{-0.72cm} \textbf{Parameters:} $ \epsilon_0>0$, $c > 1$
		\State Initialize outer loop index $t=0$
        \Do
                \State Initialize $r=1$, $\calT^{(r)} = \emptyset$, $\calV^{(r)}=\calV$ 
    		\While{$\calV^{(r)} \neq \emptyset$}
                    \State $(k^*,j^*)$ using \eqref{eq:greedy_update} with $\epsilon=\epsilon_t$
        		\State $\calT^{(r+1)}=\calT^{(r)}\cup\{(k^*,j^*)\}$
                    \State $\calS =(\calS_0,\calS_1,\ldots,\calS_{n-1})= \calS(\calT^{(r+1)})$
                    \If{ $\vert \calS_{k^*}\vert= s$}
        			\State $\calV^{(r+1)}= \calV^{(r)}\setminus \{(k^*,j), \; j=1,2,\ldots,m\}$
           \Else
           \State $\calV^{(r+1)}= \calV^{(r)}\setminus \{(k^*,j^*)\}$
        			\EndIf
        		\State $r \leftarrow r+1$
    		\EndWhile
    		\State $t \leftarrow t+1$, \;$\epsilon_{t+1} = \epsilon_t/c$
		\doWhile{$\rank{\matW_\calS}<n$}
        %$\trace{(\matW_\calS + \epsilon_{t} \matI)^{-1}} {>} 1/\epsilon_t$
		\Ensure Actuator schedule $\calS$
	\end{algorithmic}
	\label{alg:greedy_algo_epsilon}
\end{algorithm}

The intuition behind the greedy algorithm is as follows. For a given value of $\epsilon$, let $\calS^{(r)}$ be the actuator schedule obtained in the $r$th iteration of the greedy algorithm. Suppose that $\rank{\matW_{\calS^{(r)}}}=R$, which implies it has $R$ nonzero eigenvalues $\{\lambda_i\}_{i=1}^{R}$. Then, for a small $\epsilon$, the objective function is
\begin{align}
	\trace{(\matW_{\calS^{(r)}} + \epsilon \matI)^{-1}} &= \sum_{i=1}^{R} \frac{1}{\lambda_i+\epsilon} + {\frac{n-R}{\epsilon}}. %\\
% &\approx \sum_{i=1}^{R} \cfrac{1}{\lambda_i} + \cfrac{n-R}{\epsilon}.\label{eq:approx_obj_fn}
\end{align}
% The term $\sum_i 1/\lambda_i$ in \eqref{eq:approx_obj_fn} represents the average energy required to drive the system from $\vecx(0)=\zero$ to the unit norm projection of $\vecx_f$ onto the range space of
% effective controllability matrix $\matW_{\calS^{(p)}}$.
The term $(n - R)/\epsilon$ acts as the penalty added to the objective function when $\matW_{\calS^{(r)}}$ is rank deficient. 
%Therefore, the objective function in \eqref{eq:opt_energy_epsilon} forces the greedy algorithm to improve the rank of $\matW_{\calS^{(t+1)}}$ in the next iteration while minimizing the energy associated with the control. 
Our next result shows that the greedy algorithm improves the rank of the resulting Grammian until the rank becomes $n$ under the availability of such actuators for sufficiently small $\epsilon$.
% \begin{lemma} \label{prop:req_cond_to_halt}
%     If in the first $n$ steps of the greedy procedure there exists an actuator that improves the rank then there exists sufficiently small $\eps_t$ for which the greedy algorithm finds an $s$-sparse schedule $\calS_t$ with \rank($\matW_{\calS_t}$) = $n$ in the first $n$ steps of the greedy procedure for some $t^{\text{th}}$ iteration of \Cref{alg:greedy_algo_epsilon}. {\change I did not follow this lemma}
% \end{lemma}
\begin{prop} \label{prop:req_cond_to_halt}
    %\textcolor{red}{Don't we need to assume that the system is $s$-sparse controllable for this result to hold? This is because Algorithm 1 finds a sparse actuator schedule, and such a schedule exists only if the system is sparse controllable. Simply being controllable is not enough to guarantee an $s$-sparse schedule for which the system will be controllable.} 
    For a given outer loop iteration index $t$ of \Cref{alg:greedy_algo_epsilon}, let the $r$th (inner) iteration  start with $\calT^{(r)}$ and its search space be $\calV^{(r)}$ be as defined in \eqref{eq:search_space}. Suppose that the following set,
    \begin{equation} \label{eq:rank_improvement}
        \lc v\in\calV^{(r)}: \rank{\matW_{\calS(\calT^{(r)}\cup\{v\})}}
        >\rank{\matW_{\calS(\calT^{(r)})}}\rc,
    \end{equation}
    is nonempty. Then, there exists $\epsilon^*>0$ such that  if $\epsilon_t<\epsilon^*$, the next iteration of the greedy algorithm satisfies
    \begin{equation}
        \rank{\matW_{\calS(\calT^{(r+1)})}}=\rank{\matW_{\calS(\calT^{(r)})}}+1.
    \end{equation} Here, the schedule $\calS(\cdot)$ is defined in \eqref{eq:S2Tmapping} and the resulting Gramian is defined in \eqref{eq:W_mat_Defn}.
    \end{prop}
\begin{proof}
    See \Cref{app:req_halt_pf}.
\end{proof}
The above result shows that the greedy algorithm favors adding an element that increases the rank of the Gramian over an element that only increases the eigenvalues of $\matW_{\calS^{(r+1)}}$ without concurrently increasing its rank. \textcolor{black}{The inner iteration of \Cref{alg:greedy_algo_epsilon} performs a greedy optimization of \eqref{eq:opt_energy_epsilon_mod}. When the stopping criterion of \Cref{alg:greedy_algo_epsilon} is not satisfied, $\eps$ is reduced by a factor of $c$. From \Cref{prop:req_cond_to_halt}, reducing $\eps$ in the outer iteration eventually forces the algorithm to choose an actuator that improves the rank in the next inner iteration. Thus, when $\epsilon$ becomes small enough, the approach greedily chooses actuators to minimize the average control energy while improving controllability.}

Next, in order to characterize the near-optimality of \Cref{alg:greedy_algo_epsilon}, we introduce the notions of supermodularity and matroid constraints: %It is known that the above greedy strategy produces a near-optimal solution if the cost function is supermodular and the constraint set is a matroid, as defined below:
\begin{defn}[Supermodularity]\label{defn:submodular}
The set function $f : \calU \to \mathbb{R}$ is said to be $\alpha$-submodular if $\alpha \in \mathbb R_{+}$ is the largest number for which the following holds $\forall \calA \subseteq \calB \subseteq \calU$ and $\forall e \in \calU \setminus \calB$,
\begin{equation}
f(\calA \cup \{e\} ) - f(\calA) \geq {\alpha [f(\calB \cup \{e\} ) - f(\calB)]}.
\end{equation}
Also, the function $-f(\cdot)$ is said to be $\alpha$-supermodular. %If $\alpha=1$, the function $f(\cdot)$ is called submodular.
\end{defn}
\begin{defn}[Matroid]\label{defn:matroid}
A matroid is a pair $(\calU, \calI)$ where $\calU$ is a finite set and $\calI \subseteq 2^\calU$ satisfies the following three properties:
(i) $\emptyset \in \calI$; (ii) For any two sets, $\calA \subseteq \calB \subseteq \calU$, if $\calB \in \calI$, then $\calA \in \calI$; (iii) For any two sets, $\calA, \calB \in \calI$, if $\vert \calB \vert > \vert \calA \vert$, then there exists $e \in \calB \setminus \calA$ such that $\calA \cup \{e\} \in \calI$.
\end{defn}

Our next result establishes that the cost function of the optimization problem in \eqref{eq:opt_energy_epsilon_mod} is supermodular and that its constraint set is a matroid.
\begin{prop} \label{prop:submodular_proof}
The objective function of \eqref{eq:opt_energy_epsilon_mod}, $E(\calT,\epsilon) \triangleq \trace { (\matW_{\calS(\calT)} + \epsilon \matI)^{-1} }$ is an $\alpha$-supermodular function with $\alpha$ satisfying 
%\textcolor{red}{The proof shows it is supermodular, not submodular. The same confusion is in Theorem 1: the statement says supermodular but the proof says submodular.}
 \begin{equation}
 \alpha(\epsilon) \geq \cfrac{\epsilon}{\lambda_{\max}(\epsilon \matI + \matW )},\label{eq:submodular_proof}
 \end{equation}
 where $\matW \triangleq \sum_{k=1}^{n} \matA^{k-1} \matB \matB^{\T} (\matA^{k-1})^{\T}$ and $\lambda_{\max}(\cdot)$ is the largest magnitude eigenvalue. Also, \textcolor{black}{($\calV$,$\lc\calT:\;\calS(\calT)\in\Phi \rc$) is a matroid.} %Also, the constraint set $\lc\calT:\;\calS(\calT)\in\Phi \rc$ in \eqref{eq:opt_energy_epsilon_mod} is a matroid $(\calU, \calI)$ with $\calU = \calS(\calT)$, $\calT \in 2^\calV$ and $\calI = \Phi$. \textcolor{red}{What is $\matW$? Is what I have written correct? Also, from the definition above, you need two quantities, $\mathcal{U}$ and $\mathcal{I}$, in order to define a matriod. What does it mean to say that a single set, the constraint set, is a matroid? See if you agree with my modification above.}
% \textcolor{red}{This does not make sense. In the definition of a matroid, $\calI \subset 2^\calU$. This is not satisfied by the sets you are writing in blue.}
\end{prop} 
\begin{proof}
See \Cref{app:submodular_proof}.
\end{proof}

In \eqref{eq:submodular_proof}, we explicitly retain the dependence of $\alpha$ on $\eps$, as the value of $\eps$ varies across the iterations of \Cref{alg:greedy_algo_epsilon}. From \Cref{prop:submodular_proof}, we have the following guarantee for the inner loop of  \Cref{alg:greedy_algo_epsilon}.
\begin{theorem}
\label{thm:Act_sch_guarantee}
    Let $(\matA\in\bbR^{n\times m},\matB\in\bbR^{n\times m})$ be an $s$-sparse controllable LDS. Let $\beta(\epsilon) = \min \lb \frac{\alpha(\epsilon)}{2}, \frac{\alpha(\epsilon)}{1+\alpha(\epsilon)} \rb$ 
    %\textcolor{red}{(do you need $\epsilon_t$ here? Can we define it with $\epsilon$? We are saying $\epsilon=\epsilon_t$ later in this sentence.)} 
    where $\alpha(\cdot)$ is the submodularity constant of the cost function of \eqref{eq:opt_energy_epsilon_mod}. 
    %with $\epsilon=\epsilon_t$. 
    The \textcolor{black}{cost function corresponding to the set $\calS$ returned by the inner loop of \Cref{alg:greedy_algo_epsilon} run with a given value of $\epsilon$ satisfies} %$T$ is bounded by
% \begin{equation}\label{eq:Act_sch_guarantee_1}
% T \hspace{-0.75mm} \leq \hspace{-1.5mm} \left\lceil \hspace{-0.5mm} \cfrac{\log{\epsilon_0\ls n\Tilde{E}(\epsilon_{T-2}) + \beta(\epsilon_{T-2})\ls E^* - n\Tilde{E}(\epsilon_{T-2}) \rs \rs}}{\log{c}} \right\rceil \hspace{-0.5mm}+\hspace{-0.5mm}1
% \end{equation}
 % Here, $c > 1$ and $\epsilon_0>0$ are the parameters of the algorithm.
\begin{equation}\label{eq:Act_sch_guarantee_2}
\trace{\hspace{-0.5mm}(\matW_\calS + \epsilon \matI)^{-1}\hspace{-0.5mm}} < \ls 1-\beta(\epsilon) \rs \frac{n}{\epsilon} + \beta(\epsilon) E^*,
\end{equation} 
where $E^*$ is the optimal value of the objective function of the optimization problem in \eqref{eq:opt_energy_epsilon_mod} with $\epsilon=0$. %Also, when $\eps < \eps^*$ as described in \Cref{prop:req_cond_to_halt} we have $\rank{\matW_\calS} \geq \frac{n}{2}$.  %, where $\alpha$ is the submodularity constant of function $E$ given in \Cref{prop:submodular_proof} when $\epsilon=\epsilon_T$.
%{\change We can remove the last part if the stopping criteria is changed.} \textcolor{red}{I have difficulty understanding what the Theorem is saying, and what you mean by changing the stopping criterion.}
\end{theorem} 
\begin{proof}
See \Cref{app:act_schd_pf}.
\end{proof}

%We note that if the LDS is not sparse controllable, $E^*$ is not finite, and $\lambda^*$ does not exist. Consequently, we cannot guarantee that there exists an $\epsilon$ for which the stopping criteria in Step 17 of \Cref{alg:greedy_algo_epsilon} is met, and the algorithm may not stop after finite iterations. %\textcolor{blue}{ The effect of $c$ on the suboptimality can be observed by using \eqref{eq:Act_sch_guarantee_1} in the first term of \eqref{eq:Act_sch_guarantee_2} as $\eps_{T-1} = \eps_0/c^{T-1}$. We note that as $c$ approaches $1$, the suboptimality gap decreases between $\epsilon$-auxiliary energy and the optimal energy $E^*$. So, we choose $c$ close to 1, but it can lead to more iterations because the value of $\epsilon_t$ changes slowly with $t$}.

\Cref{thm:Act_sch_guarantee} shows that the cost function evaluated at the solution returned by \Cref{alg:greedy_algo_epsilon} for a given value of $\eps$ is an additive term away from the optimal cost. Admittedly, the $n/\eps$ dependence in the additive term on the right hand side makes the bound loose. In practice, \Cref{alg:greedy_algo_epsilon} returns a schedule that far outperforms, for example, a random schedule. However, besides the above result, it is hard to quantify the optimality gap even numerically,  due to the combinatorial nature of the problem. 

%\textcolor{red}{Why did we present the above theorem? Provides an upper bound on the optimality gap. It is loose, but in practice, \textcolor{blue}{the \Cref{alg:greedy_algo_epsilon}} the greedy algo returns a schedule that is close to the optimal one.} 

%\textcolor{blue}{There might exist some pathological cases where the \Cref{alg:greedy_algo_epsilon} may not terminate. As the sparsity $s$ increases from $n-\rank{\matA}$ \Cref{alg:greedy_algo_epsilon} terminates in finite iterations in practice. % ; this follows from \Cref{lma:controllabl_constr_sat}} 
When \Cref{alg:greedy_algo_epsilon} terminates, the resulting actuator schedule $\calS$ satisfies $\rank{\matW_\calS} = n$.  Note that, due to the greedy nature of the algorithm, $\calV^{(r)}$ could become an empty set before $\matW_{\calS}$ attains full rank, and the algorithm may not terminate. However, we have never found this to be an issue in any of our experiments. %The algorithm always terminates and returns a schedule $\calS$ such that $\rank{\matW_\calS} = n$.
%\textcolor{red}{Elaborate on this.} 

Finally, we note that our approach can also be used to find a time-invariant sparse actuator schedule. For this, we change $\Phi$, the feasible set of \eqref{eq:opt_energy_epsilon}, defined in \eqref{eq:feasible_set}, to $\lc(\calS_0,\calS_0,\ldots,\calS_{0})\!\!:\calS_0\!\subseteq \!\{1,2,\ldots,m\},\;\vert \calS_0 \vert \leq s\rc$, i.e., $\calS_k=\calS_0$. We can drop the index $t$ to replace $\calV$ with $\{1,2,\ldots,m\}$ and define the schedule obtained from $\calT$ as $(\calT,\calT,\ldots,\calT)$ instead of \eqref{eq:S2Tmapping} (Step 6 of \Cref{alg:greedy_algo_epsilon}). Also, in this case, $\calV^{(r)}=\{1,2,\ldots,m\}\setminus\calT^{(r)}$. Similar to \Cref{prop:submodular_proof}, we can prove that this new problem optimizes an $\alpha$-supermodular function subject to a matroid constraint, and like \Cref{thm:Act_sch_guarantee}, the modified greedy strategy possesses similar guarantees. We omit the details to avoid repetition.

\begin{figure}[t]
\minipage{0.48\textwidth}
	\includegraphics[width =1\linewidth]{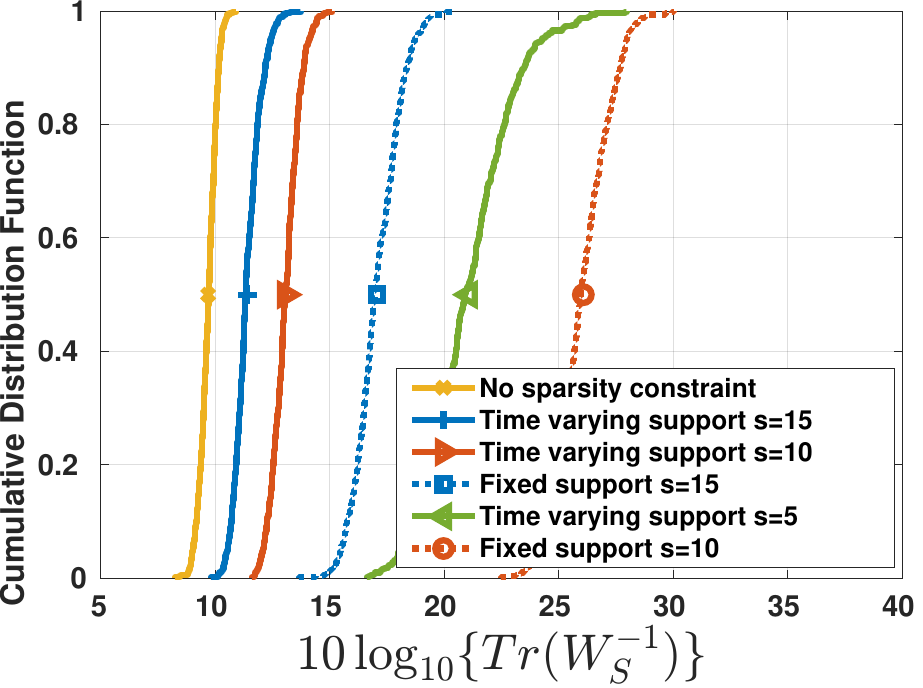}
	\caption{CDF of $10\log_{10} \lb\trace{\lb\matW_\calS \rb^{-1}}\rb$ for varying sparsity levels with $m=n=20$.} 
	\label{fig:TrCDF}
\endminipage\hfill
\end{figure}

\begin{figure}[t]
\minipage{0.48\textwidth}
	\includegraphics[width = 1\linewidth]{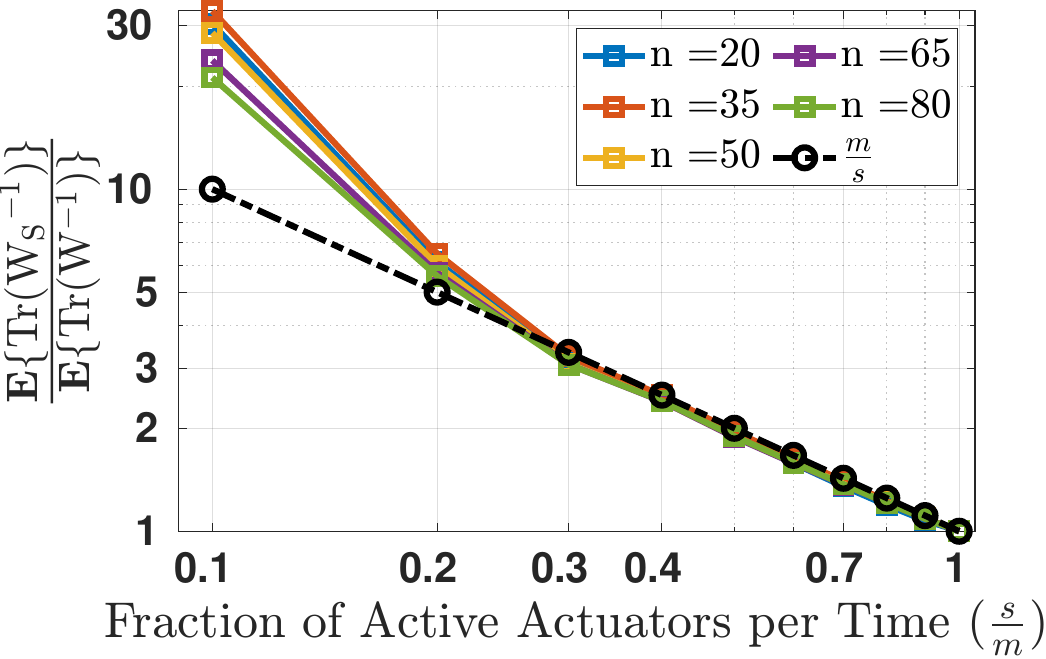}
	\caption{$\frac{\bbE_{A,B} \trace{\matW_\calS^{-1}}}{\bbE_{A,B} \trace{\matW^{-1}}} $ vs. the fraction of  active actuators ($\frac{s}{m}$),  averaged over 100 independent trials with $m=50$.} 
	\label{fig:InverseEnergy}
\endminipage\hfill
\end{figure}

\section{Simulation Results}
\label{sec:simulations}

In this section, we illustrate the performance of the time-varying and sparse actuator scheduling algorithm developed in this work by applying it to an LDS whose transfer matrix $\matA$ is generated from an Erd\H{o}s-Renyi random graph. An Erd\H{o}s-Renyi random graph consists of $n$ vertices. The pairs of vertices are independently connected by an edge with probability $p = \frac{2 \log n}{n}$, a setting used to model real-time networked systems~\cite{Tzoumas_2016_CtrlEffort}. The transfer matrix of the LDS is then generated from the graph as $\matA = \matI - \frac{1}{n} \matL$, where $\matL$ is the graph Laplacian.

We first illustrate the dependence of the average control energy, $\trace{\lb\matW_\calS \rb^{-1}}$, on the sparsity level $s$ of the control inputs. To this end, we generate the transfer matrices $\matA$ from $500$ Erd\H{o}s-Renyi random graphs drawn independently. We use \Cref{alg:greedy_algo_epsilon} to determine the sparse actuator schedule, and compute  $\trace{\lb\matW_\calS \rb^{-1}}$ corresponding to the schedule $\calS$ returned by the algorithm. For this experiment, we set $n=20$, $m=n$, and $\matB = \matI$. In \Cref{fig:TrCDF}, we show the empirical cumulative distribution function (CDF) of $\trace{\lb\matW_\calS \rb^{-1}}$ in two settings: the time-varying support case with sparsity levels $s=5, 10, 15$, and the fixed support case with $s=10, 15$. We also plot the CDF for the case with no sparsity constraint. The curves shift to the right as $s$ is decreased as expected since the system becomes more constrained, thereby incurs a higher average control energy. 
Similarly, using a fixed support makes the system more constrained and hence requires higher average control energy compared to the time-varying support case. 
Further, the control energy with time-varying support and $s=15, s=10$ is close to the control energy without sparsity constraints. 
Therefore, time-varying sparsity constraints do not impose a significant energy burden to control the LDS at moderate sparsity levels; we elaborate on this in our next experiment. 
%However, as mentioned in \Cref{sec:intro}, sparse inputs capture the natural constraints in several applications, and in particular, for networked control systems, they  save bandwidth due to their compressibility.

{
Next, we illustrate the relative cost of imposing the sparsity constraints compared to the non-sparse case. Here, the entries of the input matrix $\matB$ are i.i.d. and uniformly distributed over $[0,1]$. We plot $\rho \triangleq {\bbE_{A,B} \trace{\matW_\calS^{-1}}}/{\bbE_{A,B} \trace{\matW^{-1}}}$ as a function of the fraction of active actuators at each time instant (${s}/{m}$). The quantity $\rho$ represents the ratio of the average control energy with the sparsity constraint to the average control energy in the unconstrained case. We show the behavior in \Cref{fig:InverseEnergy}, with $m$ fixed at $50$ and for various values of $n$ from $20$ to $80$. 
%we investigate the effect of choosing various sparsity levels ($s$) on the $\trace{\matW_\calS^{-1}}$ for different systems (varying $n$) with fixed $m$. In our setup, we fix $m = 50$ and vary $n$ from $20$ to $80$. 
%is a random matrix \textcolor{red}{(what is a uniform random matrix? This is new terminology to me.)} with i.i.d. entries \textcolor{red}{drawn from \textcolor{blue}{uniform} distribution}. 
We observe that there is nearly an inverse-linear relationship between $s/m$ and $\rho$. This shows that using time-varying $s$-sparse control inputs increases the average energy to control the system by a factor proportional to the reciprocal of the fraction of active actuators. We have also observed this trend when $\matB$ is a random matrix with i.i.d. entries drawn from a Gaussian distribution and when $\matB$ = $\matI$ (for $n=m$). With reasonable approximations, it is possible to mathematically derive that this behavior holds when $\matB$ has i.i.d. entries; we omit the details due to lack of space.
%A rough mathematical justification is provided in \Cref{app:fig:InverseEnergy_pf}.}

\begin{figure}[t]
\minipage{0.48\textwidth}
	\includegraphics[width = 1\linewidth]{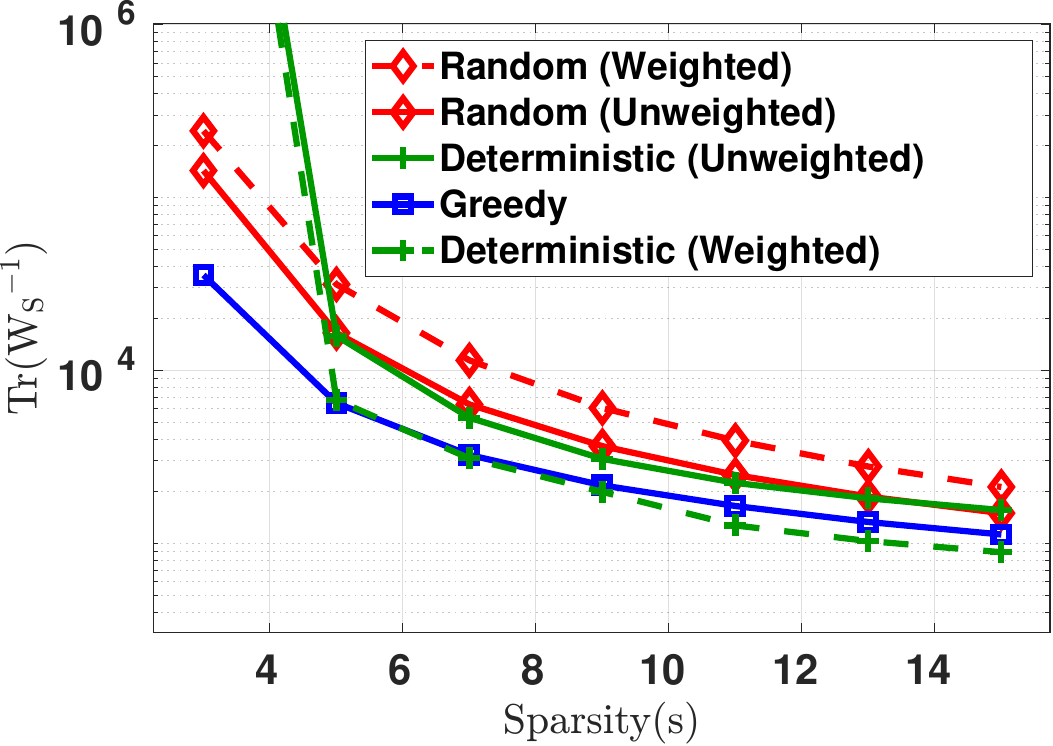}
	\caption{$\trace{\matW_\calS^{-1}}$ as a function of sparsity ($s$), averaged over $100$ independent trials with $n=100, m=50$.} 
	\label{fig:ActComparision}
\endminipage\hfill
\end{figure}

Next, we compare \Cref{alg:greedy_algo_epsilon} with schedulers that are adapted from~\cite{jadbabaie2018deterministic}, the closest to our work. The algorithms in \cite{jadbabaie2018deterministic} are developed under an \emph{average} sparsity constraint, i.e., the individual inputs need not be sparse, but the overall sequence of inputs need to satisfy an average sparsity constraint. We modify the algorithms in \cite{jadbabaie2018deterministic} by adding a piecewise sparsity constraint. In \Cref{fig:ActComparision}, we plot  $\trace{\matW_{\calS}^{-1}}$ as a function of $s$ for four algorithms adapted from \cite{jadbabaie2018deterministic}, namely, random-weighted, random-unweighted, deterministic-unweighted, and deterministic-weighted. \textcolor{black}{The random scheduler samples an actuator from the probability distribution given in \cite[Algorithm 6]{jadbabaie2018deterministic} and adds it to the schedule provided adding it still satisfies the sparsity constraint. The deterministic scheduler picks a actuator that greedily optimizes the objective function used in \cite[Algorithm 1]{jadbabaie2018deterministic}. The weighted schedulers have an additional weight (amplification) associated with the selected actuator, while the weights are set to $1$ for the unweighted schedule (see \cite[Algorithms 2, 6, 7]{jadbabaie2018deterministic}.) 
%Specifically, we implemented Algorithms 2, 6, 7 from~\cite{jadbabaie2018deterministic}. 
For the unweighted random scheduler, the actuators are sampled without replacement. The plot shows that \Cref{alg:greedy_algo_epsilon} generally outperforms existing algorithms modified to satisfy our sparsity constraint, especially at lower sparsity levels.} 
%At higher sparsity levels ($s > 10$), the deterministic-weighted schedule slightly outperforms \Cref{alg:greedy_algo_epsilon}.}

%\textcolor{red}{What iterations are you referring to here? If it is because Algorithm 2 of [4] is iterative, then explicitly mention which algorithm you are referring to. Overall, it is very confusing. Also, why do I need this information, i.e., what is the reason you are mentioning the number of iterations of the algorithm?} In low sparsity ($s$) regime, \textcolor{red}{Again, why do I (as a reader) need this information?} we see that proposed greedy algorithm has lower average energy than the unweighted scheduling algorithms, and perform similar to the weighted deterministic actuator scheduler adapted from Algorithm 2 of \cite{jadbabaie2018deterministic}. The unweighted random sampling has better performance than the random weighted schedule because the actuators are chosen without replacement, resulting in larger support size than the weighted schedule. The performance of all the actuator schedulers are similar at higher sparsity ($s$). The weighted deterministic schedule has slightly better performance for different $m$ relative to the proposed greedy algorithm. \textcolor{red}{Need to significantly improve the discussion here. It reads like a set of scattered comments that is hard to make sense out of.}}

\section{Conclusion}
We studied the problem of finding an energy-efficient actuator schedule for controlling an LDS in a finite number of steps, with at most $s$ active inputs per time step. We presented a greedy scheduling algorithm by considering the minimization of $\trace{\matW_\calS^{-1}}$, which models the average control energy. We presented several interesting theoretical guarantees for the algorithm. Through simulations, we showed that the algorithm returns a feasible schedule, with the average energy increasing (compared to the case without sparsity constraints) by a factor proportional to the inverse of the fraction of active inputs, for randomly drawn LDSs. Overall, we conclude that sparse inputs can drive an LDS to a desired state without a significant increase in energy requirements, at moderate sparsity levels.

\appendices
\crefalias{section}{appendix}

\section{Proof of \Cref{prop:req_cond_to_halt}}
\label{app:req_halt_pf}
Let $\calS$ be the actuator schedule at some iteration of the greedy algorithm such that $\rank{\matW_{\calS}}$ = $R < n$. By \eqref{eq:rank_improvement}, there exists a candidate feasible schedule $\hat{\calS}$ for the next iteration such that $\rank{\matW_{\hat{\calS}}}$ = $R+1$.
%Suppose there exists a feasible actuator such that rank is improved by adding it to the schedule $\calS$. Let there be 
Suppose there is an alternative feasible schedule $\Tilde{\calS}$ such that $\rank{\matW_{\Tilde{\calS}}}$ = $R$.  The greedy algorithm chooses the schedule $\hat{\calS}$ if \[ \trace{(\matW_{\Tilde{\calS}}+\eps\matI)^{-1}} > \trace{(\matW_{\hat{\calS}}+\eps\matI)^{-1}}. \]
Let $\lambda_i, i = 1, 2, \ldots n$, with $\lambda_1 \geq \lambda_2 \geq \ldots \geq \lambda_n$ denote the ordered eigenvalues of $\matW_\calS$ (similarly, $\hat{\lambda}_i$ and $\tilde{\lambda}_i$ \textcolor{black}{for $\matW_{\hat{S}}$ and $\matW_{\tilde{S}}$, respectively}). Then, we can write
\begin{gather}
    \trace{(\matW_{\Tilde{\calS}}+\eps\matI)^{-1}} \geq \frac{R}{\Tilde{\lambda}_1 + \eps} + \frac{n-R}{\eps}, \\
    \trace{(\matW_{\hat{\calS}}+\eps\matI)^{-1}} \leq \frac{R+1}{\hat{\lambda}_{R+1} + \eps} + \frac{n-R-1}{\eps}.
\end{gather}
The following inequality is a sufficient condition for the greedy algorithm to choose $\hat{\calS}$:
\begin{equation}
    \frac{R}{\Tilde{\lambda}_1 + \eps} + \frac{n-R}{\eps} > \frac{R+1}{\hat{\lambda}_{R+1} + \eps} + \frac{n-R-1}{\eps},
\end{equation}
Simplifying the above, we get
\begin{equation}
    \Tilde{\lambda}_1 \hat{\lambda}_{R+1} > \eps \ls R \Tilde{\lambda}_1 - (R+1) \hat{\lambda}_{R+1} \rs, 
\end{equation}
The inequality holds trivially for $R \Tilde{\lambda}_1 - (R+1) \hat{\lambda}_{R+1} \leq 0$.
    For $R \Tilde{\lambda}_1 - (R+1) \hat{\lambda}_{R+1} > 0$, we obtain
\begin{equation} \label{eq:eps_ineq}
    \eps < \frac{\Tilde{\lambda}_1 \hat{\lambda}_{R+1}}{R \Tilde{\lambda}_1 - (R+1) \hat{\lambda}_{R+1}}.
\end{equation} 
The right hand side of above equation is finite and positive, and $\eps$ is decreased by factor $c>1$ in each iteration of the \Cref{alg:greedy_algo_epsilon}. Hence, there will always exists an $\eps$ for which \eqref{eq:eps_ineq} is satisfied. Hence, for a sufficiently small $\eps$, \eqref{eq:eps_ineq} holds for any $\Tilde{\calS}$ given an $\hat{\calS}$. Thus, the greedy algorithm chooses an actuator that improves the rank. 
%\textcolor{red}{I am not convinced by the proof. I agree that if there are two schedules, of which improves the rank and the other doesn't, Algorithm 1 will pick the one that improves the rank provided $\epsilon$ is small enough. However, it is not clear that you will be able to make rank Ws = n eventually, because for an arbitrary initial choice $\calS$, it is not clear that there will be a candidate $\hat{\calS}$ that improves the rank.}
\hfill \qedb

\section{Proof of \Cref{prop:submodular_proof}}\label{app:submodular_proof}
From \cite[Theorem 2]{Chamon_2017_MSESupermodular}, we know that  
%which states that any set $\calV$, 
the real-valued function $f$ of the form %\textcolor{red}{the sentence is poorly constructed and does not make sense. Also the equation below introduces new notation that is hard to parse through.}
\begin{equation}
f({\calA}) = \trace{\lb\matM_{\emptyset} + \sum_{i \in \calA} \matM_i \rb^{-1} },
\end{equation}
where $\calA$ is an index set and  $\matM_{\emptyset} \succ 0 $ and $\matM_i \succeq 0$, is monotonically decreasing and $\alpha$-supermodular with $\alpha  \geq \frac{\mu_{\min}}{\mu_{\max}} > 0,$
where
\begin{equation}
0 < \mu_{\min} \leq \lambda_{\min}(\matM_{\emptyset}) \leq \lambda_{\max}(\matM_{\emptyset} + \sum_{i \in \calA} \matM_i) \leq \mu_{\max}.
\end{equation}
Here, $\lambda_{\max}(\cdot)$ and $\lambda_{\min}(\cdot)$ are the largest and smallest eigenvalues in magnitude, respectively. Our cost function $E(\calT, \epsilon)$ takes the same form as $f(\calA)$,  with $\matM_{\emptyset} = \epsilon \matI\succ 0$ and the matrices $\matM_i$ being constructed as the submatrices of the controllability gramian $\matW$ (see \eqref{eq:W_mat_Defn}). Further, from the definitions, we can set $\mu_{\min} = \epsilon$ and $\mu_{\max} = \lambda_{\max}(\epsilon \matI + \matW)$. Hence, the function $E(\calT, \epsilon)$ is $\alpha$-supermodular with $\alpha$ satisfying \eqref{eq:submodular_proof}. %\textcolor{red}{There is some confusion between ``submodular'' and ``supermodular'' in this proof. The result from [21] says $f$ is supermodular but in the end you are claiming that our cost function is submodular.}

We next complete the proof by showing that the sparsity constraint on the actuator set is a matroid constraint. For this proof, we consider the equivalent sparsity constraint on $\calT\in2^\calV$ in \eqref{eq:opt_energy_epsilon_mod}. The constraint is given by
\begin{equation}\label{eq:sparse_T}
\vert\{j: (k,j)\in\calT\}\vert\leq s, \;\text{for } k=0,1,\ldots,n-1
\end{equation}
To prove that the above constraint is a matroid, we verify the three conditions of \Cref{defn:matroid}. The first condition holds trivially. If we have a set $\calT$ that satisfies the sparsity constraint in \eqref{eq:sparse_T}, then any subset $\calT' \subseteq \calT$ also satisfies the sparsity constraint. Hence, the second condition holds.
For the third condition, suppose there exist sets $\calT,\calT'\in2^{\calV}$ satisfying the sparsity constraint in \eqref{eq:sparse_T} and $\vert \calT \vert > \vert \calT' \vert$. Then, there exists at least one time index $k$ such that 
\begin{equation}
\vert\{j: (k,j)\in\calT\}\vert> \vert\{j: (k,j)\in\calT'\}\vert.
\end{equation}
If we take an element $\tilde{j}\in \{j: (k,j)\in\calT\setminus\calT'\}$,  the new set $\calT' \cup \{(k,\tilde{j})\}$ satisfies
\begin{align}
\lv\lc j: (k,j)\in\calT' \rc\cup \{(k,\tilde{j})\}\rv&= \lv\lc j: (k,j)\in\calT' \rc\rv+1 \nonumber\\
&\leq \lv\lc j: (k,j)\in\calT\rc\rv\leq s. \nonumber
\end{align}
So, the new set $\calT' \cup \{(k,\tilde{j})\}$ also satisfies \eqref{eq:sparse_T}, verifying the third condition and completing the proof.
\hfill\qedb

\section{Proof of \Cref{thm:Act_sch_guarantee}}
\label{app:act_schd_pf}
%The proof relies on the optimality of the greedy algorithm when used to solve combinatorial problems with submodular objective function and matroid constraint~\cite{9030055}. 
Consider the problem of minimizing an $\alpha$-supermodular function $f:2^\calV \rightarrow \mathbb{R}$ subject a matroid constraint $(\calV, \calI)$~ i.e.,
\begin{equation} \label{opt:submodular_int_matroid}
\min_{\calA\in 2^{\calV}} f(\calA) \;
\text{s.t.} \;  \calA \in \calI  . 
\end{equation}
Suppose $f^*$ is the optimal solution of \eqref{opt:submodular_int_matroid} and $\tilde{f}$ is the solution returned by \textcolor{black}{a greedy algorithm that starts with $\calA = \emptyset$ and adds elements from $\calV$ one-by-one, so that at step $t$, it updates $\calA_t$ as $\calA_{t+1} = \calA_t \cup e^*$, where $e^* = \text{argmin}_{e \in \calV \backslash \calA_t, \calA_t \cup e \in \calI} f(\calA_t \cup e)$. }
%of the form in \Cref{alg:greedy_algo_epsilon}.} 
Then, from \cite[Theorem 1]{Chamon_2019_MatroidOpt} we have, 
\begin{equation} \label{eq_alpha_sub_guarent}
	\frac{f^* - \tilde{f}}{f^*-f(\emptyset)} \leq 1 - \min \lb\frac{\alpha}{2} , \frac{\alpha}{1+\alpha} \rb,
\end{equation}
where $\alpha$ is the submodularity constant of $f$. %\textcolor{red}{The result in the reference does not have the $\alpha/2$ term.}

%Clearly, our optimization problem in \eqref{eq:opt_energy_epsilon_mod} takes the same form as that of \eqref{opt:submodular_int_matroid} and \Cref{alg:greedy_algo_epsilon} is a greedy algorithm. 
From \Cref{prop:submodular_proof}, our cost function is $\alpha$-supermodular and the constraint set is a matroid. Further, \Cref{alg:greedy_algo_epsilon} is a greedy algorithm of the form described above. For a fixed $\epsilon$, let $E^*(\epsilon)$ be the optimal energy obtained by solving \eqref{eq:opt_energy_epsilon_mod} exactly and $\tilde{E}(\epsilon) = \trace{\lb \matW_\calS + \epsilon \matI \rb^{-1} }$ be the objective function returned by \Cref{alg:greedy_algo_epsilon} with parameter $\epsilon$. Then, using \eqref{eq_alpha_sub_guarent}, we have
\begin{equation} \label{eq_guarantee1}
\frac{\tilde{E}(\epsilon) - E^*(\epsilon)}{n/\epsilon - E^*(\epsilon)} \leq 1 - \beta(\epsilon),
\end{equation}
where $\beta(\eps)>0$ is defined in \Cref{thm:Act_sch_guarantee}. Let, $E^* = \trace{\matW_{\calS^*}^{-1}}$ where $\calS^*$ be the corresponding optimal actuator schedule obtained by solving \eqref{eq:opt_energy_epsilon_mod} with $\eps=0$. We obtain the following inequality,
\begin{equation} \label{eq:estarineq}
	E^* = \trace{\matW_{\calS^*}^{-1} } > \trace{\lb \matW_{\calS^*}+\epsilon\eye \rb^{-1} } \geq E^*(\epsilon).
\end{equation}
Plugging in this inequality into \eqref{eq_guarantee1} leads to
\begin{equation} \label{eq_derived_guarantee1}
	\tilde{E}(\epsilon) \leq \frac{n(1-\beta(\epsilon))}{\epsilon} + \beta(\epsilon) E^*(\epsilon) < \frac{n(1-\beta(\epsilon))}{\epsilon} + \beta(\epsilon) E^*.
\end{equation}
Hence, the proof is complete. \hfill \qedb

\bibliographystyle{IEEEtran}
\bibliography{IEEEabrv,refs}
\end{document}

%% file: My_notations.tex
\newtheorem{theorem}{Theorem}

\newtheorem{prop}{Proposition}

\newtheorem{defn}{Definition}

\newcommand{\bs}{\boldsymbol}
\newcommand{\bb}{\mathbb}
\newcommand{\mcal}{\mathcal}

\newcommand{\eye}{\bs{I}}
\newcommand{\zero}{\bs{0}}

\newcommand{\lb}{\left(}
\newcommand{\rb}{\right)}
\newcommand{\ls}{\left[}
\newcommand{\rs}{\right]}
\newcommand{\lc}{\left\{}
\newcommand{\rc}{\right\}}

\newcommand{\lv}{\left\vert}
\newcommand{\rv}{\right\vert}

\newcommand{\LRV}[1]{{\left\vert\kern-0.25ex\left\vert\kern-0.25ex\left\vert #1 \right\vert\kern-0.25ex\right\vert\kern-0.25ex\right\vert}}
\newcommand{\matV}[1]{{\left\vert\kern-0.25ex\left\vert\kern-0.25ex\left\vert #1 \right\vert\kern-0.25ex\right\vert\kern-0.25ex\right\vert}}

\newcommand{\T}{\mathsf{T}}
\newcommand{\trace}[1]{\mathsf{Tr}\lc#1\rc}

\newcommand{\rank}[1]{\mathsf{Rank}\lc#1\rc}

\newcommand{\matA}{\bs{A}}
\newcommand{\matB}{\bs{B}}

\newcommand{\matI}{\bs{I}}

\newcommand{\matL}{\bs{L}}
\newcommand{\matM}{\bs{M}}

\newcommand{\matW}{\bs{W}}

\newcommand{\bbE}{\bb{E}}

\newcommand{\bbR}{\bb{R}}

\newcommand{\calA}{\mcal{A}}
\newcommand{\calB}{\mcal{B}}

\newcommand{\calI}{\mcal{I}}

\newcommand{\calS}{\mcal{S}}
\newcommand{\calT}{\mcal{T}}
\newcommand{\calU}{\mcal{U}}
\newcommand{\calV}{\mcal{V}}

\newcommand{\vecu}{\bs{u}}

\newcommand{\vecx}{\bs{x}}

\newcommand{\eps}{\epsilon}

%% file: Sparse_Actuator_Schedule_of_LDS.bbl
\begin{thebibliography}{10}
\providecommand{\url}[1]{#1}
\csname url@rmstyle\endcsname
\providecommand{\newblock}{\relax}
\providecommand{\bibinfo}[2]{#2}
\providecommand\BIBentrySTDinterwordspacing{\spaceskip=0pt\relax}
\providecommand\BIBentryALTinterwordstretchfactor{4}
\providecommand\BIBentryALTinterwordspacing{\spaceskip=\fontdimen2\font plus
\BIBentryALTinterwordstretchfactor\fontdimen3\font minus
  \fontdimen4\font\relax}
\providecommand\BIBforeignlanguage[2]{{%
\expandafter\ifx\csname l@#1\endcsname\relax
\typeout{** WARNING: IEEEtran.bst: No hyphenation pattern has been}%
\typeout{** loaded for the language `#1'. Using the pattern for}%
\typeout{** the default language instead.}%
\else
\language=\csname l@#1\endcsname
\fi
#2}}

\bibitem{control_comm_constraint}
S.~C. Tatikonda, ``Control under communication constraints,'' Ph.D.
  dissertation, Dept. Elect. Comput. Sci., MIT, Cambridge, USA, 2000.

\bibitem{jadhbabaie2019}
A.~{Jadbabaie}, A.~{Olshevsky}, G.~J. {Pappas}, and V.~{Tzoumas}, ``Minimal
  reachability is hard to approximate,'' \emph{IEEE Trans. Autom. Control},
  vol.~64, no.~2, pp. 783--789, 2019.

\bibitem{NAGAHARA201184}
M.~Nagahara and D.~E. Quevedo, ``Sparse representations for packetized
  predictive networked control,'' \emph{IFAC Proc. Volumes}, vol.~44, no.~1,
  pp. 84 -- 89, 2011.

\bibitem{Li_2016_CSinSwitchedSys}
Z.~{Li}, Y.~{Xu}, H.~{Huang}, and S.~{Misra}, ``Sparse control and compressed
  sensing in networked switched systems,'' \emph{IET Control Theory Appl.},
  vol.~10, no.~9, pp. 1078--1087, 2016.

\bibitem{foucart2013math}
S.~Foucart and H.~Rauhut, \emph{A Mathematical Introduction to Compressive
  Sensing}.\hskip 1em plus 0.5em minus 0.4em\relax Birkh\"{a}user, 2013.

\bibitem{joseph2020controllability}
G.~Joseph, B.~Nettasinghe, V.~Krishnamurthy, and P.~K. Varshney,
  ``Controllability of network opinion in {Erdos-Renyi} graphs using sparse
  control inputs,'' \emph{SIAM J. Control Optim.}, vol.~59, no.~3, pp.
  2321--2345, 2021.

\bibitem{WENDT201937}
N.~Wendt, C.~Dhal, and S.~Roy, ``Control of network opinion dynamics by a
  selfish agent with limited visibility,'' \emph{IFAC-PapersOnLine}, vol.~52,
  no.~3, pp. 37 -- 42, 2019, iFAC Symp. Large Scale Complex Syst.

\bibitem{gjoseph2020controllability}
G.~Joseph and C.~R. Murthy, ``Controllability of linear dynamical systems under
  input sparsity constraints,'' \emph{IEEE Trans. Autom. Control}, vol.~66,
  no.~2, pp. 924--931, 2021.

\bibitem{Olshvesky_2020_ActRlx}
A.~Olshevsky, ``On a relaxation of time-varying actuator placement,''
  \emph{{IEEE} Control Syst. Lett.}, vol.~4, no.~3, pp. 656--661, 2020.

\bibitem{Ikeda_2018_TVNodeSlct}
T.~Ikeda and K.~Kashima, ``Sparsity-constrained controllability maximization
  with application to time-varying control node selection,'' \emph{{IEEE}
  Control Syst. Lett.}, vol.~2, no.~3, pp. 321--326, 2018.

\bibitem{Fotiadis_2021_ActPlctLearn}
F.~Fotiadis and K.~G. Vamvoudakis, ``Learning-based actuator placement for
  uncertain systems,'' in \emph{Proc.\ {CDC}}, 2021, pp. 90--95.

\bibitem{Fotiadis_2022_DataBasedActSlct}
F.~Fotiadis, K.~G. Vamvoudakis, and Z.-P. Jiang, ``Data-based actuator
  selection for optimal control allocation,'' in \emph{Proc.\ {CDC}}, 2022, pp.
  4674--4679.

\bibitem{Olshevesky_2014_MinControl}
A.~{Olshevsky}, ``Minimal controllability problems,'' \emph{{IEEE} Trans.
  Control Netw. Syst.}, vol.~1, no.~3, pp. 249--258, 2014.

\bibitem{Tzoumas_2016_CtrlEffort}
V.~{Tzoumas}, M.~A. {Rahimian}, G.~J. {Pappas}, and A.~{Jadbabaie}, ``Minimal
  actuator placement with bounds on control effort,'' \emph{{IEEE} Trans.
  Control Netw. Syst.}, vol.~3, no.~1, pp. 67--78, 2016.

\bibitem{Dilip_2019_CtrlGramOptSlct}
A.~S.~A. Dilip, ``The controllability gramian, the hadamard product, and the
  optimal actuator/leader and sensor selection problem,'' \emph{{IEEE} Control
  Syst. Lett.}, vol.~3, no.~4, pp. 883--888, 2019.

\bibitem{jadbabaie2018deterministic}
M.~Siami, A.~Olshevsky, and A.~Jadbabaie, ``Deterministic and randomized
  actuator scheduling with guaranteed performance bounds,'' \emph{IEEE Trans.
  Autom. Control}, vol.~66, no.~4, pp. 1686--1701, 2020.

\bibitem{SIAMI2020109054}
M.~Siami and A.~Jadbabaie, ``A separation theorem for joint sensor and actuator
  scheduling with guaranteed performance bounds,'' \emph{Automatica}, vol. 119,
  p. 109054, 2020.

\bibitem{Chamon_2019_MatroidOpt}
L.~F.~O. {Chamon}, A.~{Amice}, and A.~{Ribeiro}, ``Matroid-constrained
  approximately supermodular optimization for near-optimal actuator
  scheduling,'' in \emph{Proc.\ {CDC}}, 2019, pp. 3391--3398.

\bibitem{Jiao_2023_ActAchCvxRlx}
J.~Jiao, D.~Maity, J.~S. Baras, and S.~Hirche, ``Actuator scheduling for linear
  systems: A convex relaxation approach,'' \emph{{IEEE} Control Syst. Lett.},
  vol.~7, pp. 7--12, 2023.

\bibitem{Zhao_2016_NodeSchControl}
Y.~Zhao, F.~Pasqualetti, and J.~Cortés, ``Scheduling of control nodes for
  improved network controllability,'' in \emph{Proc.\ {CDC}}, 2016, pp.
  1859--1864.

\bibitem{Olshevsky_2018_NonSupermodular}
A.~{Olshevsky}, ``On (non)supermodularity of average control energy,''
  \emph{IEEE Tran. Control Netw. Syst.}, vol.~5, no.~3, pp. 1177--1181, 2018.

\bibitem{Chamon_2017_MSESupermodular}
L.~F.~O. {Chamon}, G.~J. {Pappas}, and A.~{Ribeiro}, ``The mean square error in
  {Kalman} filtering sensor selection is approximately supermodular,'' in
  \emph{Proc.\ {CDC}}, 2017, pp. 343--350.

\end{thebibliography}
